\documentclass[11pt]{article}

\usepackage[letterpaper,margin=1.05in]{geometry}
\usepackage{amsmath,amssymb,amsthm,mathtools}
\usepackage{microtype}
\usepackage{xcolor}
\usepackage{enumitem}
\usepackage[numbers,sort&compress]{natbib}
\usepackage[colorlinks=true,linkcolor=blue,citecolor=blue,urlcolor=blue]{hyperref}
\hypersetup{
  pdftitle={Sharp Tangential NEC Minimization and Israel Surface Layers in Schwarzschild Mass Interpolations},
  pdfauthor={Changsun Choi and Ryan E. Grady}
}
\usepackage{fancyhdr}
\usepackage{tikz}
\usetikzlibrary{arrows.meta,positioning,patterns}

\allowdisplaybreaks

\newtheorem{theorem}{Theorem}[section]
\newtheorem{proposition}[theorem]{Proposition}
\newtheorem{lemma}[theorem]{Lemma}
\newtheorem{corollary}[theorem]{Corollary}
\theoremstyle{definition}
\newtheorem{definition}[theorem]{Definition}

\newtheorem{remark}[theorem]{Remark}

\newcommand{\R}{\mathbb R}
\newcommand{\Ainf}{\mathcal A_{\infty}}
\newcommand{\ABV}{\mathcal A_{BV}}
\newcommand{\I}{\mathcal I}
\newcommand{\JBV}{\mathcal J}
\newcommand{\supp}{\operatorname{supp}}
\newcommand{\diag}{\operatorname{diag}}
\newcommand{\weakstar}{\xrightarrow{\,\mathrm{w}^*\,}}
\newcommand{\one}{\mathbf 1}

\title{\textbf{{Sharp Tangential NEC Minimization and Israel Surface Layers in Schwarzschild Mass Interpolations}}}
\author{
Changsun Choi\thanks{
Department of Physics, Montana State University,
         Bozeman 59717 USA.
Email: \texttt{changsunchoi@montana.edu}
}
\and
Ryan E.~Grady\thanks{
Department of Mathematical Sciences, Montana State University,
         Bozeman 59717 USA.
Email: \texttt{ryan.grady1@montana.edu}
}
}
\date{}

\begin{document}
\maketitle

\begin{abstract}
We study static, spherically symmetric metrics in Schwarzschild gauge whose mass function increases smoothly from $M_1$ to $M_2>M_1$ across an annulus $[R_1,R_2]$ outside the larger Schwarzschild radius.  {An elementary obstruction shows that tangential NEC violation is unavoidable in this class. This reduces the physical question to a quantitative one: what is the least possible violation, and what geometry is selected by near-minimization?}  We first determine the exact $L^1$-relaxation of the resulting nonlocal weighted positive-variation functional and solve the relaxed problem explicitly.  Its unique minimizer is a normalized box profile, {which yields the infimum of violations: the relaxed minimum is attained, while the smooth infimum is not.}  An exact deficit decomposition yields sharp quantitative stability, with a square-root rate at a nondegenerate critical optimizer and a linear rate at a strict constrained boundary optimizer; the same rates control the mass function and metric coefficients.  We identify a variational phase transition between these regimes and prove reciprocal-width divergence in the thin-annulus limit.  {Every smooth minimizing sequence converges to one locally Lipschitz metric with a
constant-density $p=-\rho$ bulk and two timelike surface layers, its Einstein
tensors converge distributionally, and the singular terms agree with the Israel
surface stress tensors.} The negative tangential null energy concentrates on the inner layer, whose integrated negative pressure equals the sharp variational cost.
\end{abstract}

\noindent\textbf{Keywords.} Null energy condition; no--go theorem; gravastar; relaxation; functions of bounded variation; quantitative stability; Schwarzschild mass function; distributional curvature; Israel junction conditions.

\noindent\textbf{Mathematics Subject Classification.} 83C57, 49J45, 49Q20, 83C75.

\section{Introduction}

Energy conditions convert algebraic restrictions on the stress-energy tensor into geometric and causal information.  When a prescribed spacetime construction cannot avoid violating an energy condition, a natural next question is quantitative: what is the least violation compatible with the construction, and what geometry is approached when that lower bound is nearly attained?  Volume-integral measures of energy-condition violation have been considered in the wormhole literature, for example in \cite{VisserKarDadhich2003,KarDadhichVisser2004}.  The present paper studies a different exactly solvable problem: a monotone increase of Schwarzschild mass across a fixed annulus in Schwarzschild gauge.

{The analytic framework of $BV$ relaxation is standard \cite{AmbrosioFuscoPallara2000,EvansGariepy2015}; hypersurface-supported curvature and the Israel junction formalism are classical \cite{Israel1966,GerochTraschen1987,MarsSenovilla1993,MansouriKhorrami1996,LeFlochMardare2007}.  Gravastar models likewise motivate configurations with a $p=-\rho$ region and thin surface layers \cite{MazurMottola2004,VisserWiltshire2004}.  The main theorem below states the paper's contribution: an exact solution and relaxation of the present nonlocal variational problem, sharp rigidity of its near-minimizers, and sequence-independent selection of
a bulk--two-layer Israel geometry of gravastar type in the sense made precise in
Remark~\ref{rem:gravastar}. Here “gravastar-type” refers only to the presence of a $p=-\rho$ region bounded by surface layers. Unlike the standard gravastar construction, which typically joins a de Sitter interior to a Schwarzschild exterior through a single transition shell, our limiting geometry contains a constant-density annulus between two Schwarzschild vacuum regions and has two surface layers selected by the variational problem.
}

Fix
\[
0<M_1<M_2,
\qquad
\Delta M:=M_2-M_1,
\qquad
R_2>R_1>2M_2.
\]
We assume $R_2>R_1>2M_2$ because this makes the induced metric static, as we will see. For a nonnegative smooth profile $u$ supported in $[R_1,R_2]$ and normalized by
\[
\int_{R_1}^{R_2}r^2u(r)\,dr=1,
\]
define
\[
P_u(r):=\int_{R_1}^r s^2u(s)\,ds,
\qquad
m_u(r):=M_1+\Delta M\,P_u(r),
\]
and
\[
F_u(r):=1-\frac{2m_u(r)}r.
\]
The corresponding metric, in geometric units $G=c=1$, is
\[
g_u=-F_u(r)\,dt^2+F_u(r)^{-1}\,dr^2+r^2d\Omega^2.
\]
Within this ansatz one has $p_r=-\rho$, while
\[
\rho+p_t=-\frac{\Delta M}{8\pi}ru'.
\]
{Because every admissible profile starts and ends at zero but has positive weighted integral, it must rise on a set of positive measure.  Hence $\rho+p_t<0$ somewhere for every smooth interpolation.  Proposition~\ref{thm:nogo} records this no--go statement and a profile-independent positive lower bound; it is the motivation for replacing the qualitative question by a sharp minimization problem.}
Consequently, after removing the fixed physical factor $\Delta M/2$, the proper-volume-integrated negative tangential NEC is
\[
\I[u]
=
\int_{R_1}^{R_2}
\frac{r^3(u'(r))_+}{\sqrt{F_u(r)}}\,dr.
\]
The problem is genuinely nonlocal because the weight $F_u^{-1/2}$ depends on the cumulative profile $P_u$.

The first main contribution is analytical.  We prove that the canonical $L^1$-relaxation of $\I$ to the natural $BV$ class is
\[
\JBV[u]
=
\int_{[R_1,R_2]}
\frac{r^3}{\sqrt{F_u(r)}}\,d(Du)^+(r).
\]  Notice that $\I[u]$ contains a Riemann integral, whereas $\JBV[u]$ features an integral with respect to a measure.
{The relaxation is the canonical lower-semicontinuous completion of the smooth
problem: for an idealized BV profile $u$, $\mathcal{J}[u]$ is the least
asymptotic tangential-NEC cost among smooth matter layers converging to $u$ in 
$L^1$, including energy concentrated in thin transition zones.}  The relaxed functional has a unique minimizer, a normalized box profile.

The second contribution is rigidity.  An exact decomposition of the energy deficit separates three losses: positive variation after mass has accumulated, positive variation away from the optimal radius, and negative variation before the outer endpoint.  This decomposition gives convergence of every smooth minimizing sequence and sharp stability estimates. The stability theory has two regimes: a square-root rate when the optimizer is the freely selected interior point $a=x_0$, and a linear rate when the optimizer is pinned to the boundary $a=R_1>x_0$. The interior regime is possible only in the restricted low-mass-contrast range
\[
M_2<\frac{7M_1}{6},
\qquad\text{equivalently}\qquad
\frac{\Delta M}{M_1}<\frac16,
\]
because $R_1>2M_2$ and $x_0<7M_1/3$. Outside this range, only the boundary phase and the linear stability regime can occur.
 Almost minimal NEC cost therefore forces the matter profile, mass function, and metric coefficients to be close to one specific geometry.

The third contribution is geometric and physical.  The Einstein tensors of every smooth minimizing sequence converge as tensor distributions.  The limit consists of a constant-density bulk region and two timelike surface layers.  Their Dirac coefficients agree exactly with the Israel tensors, and the negative part is concentrated on the inner layer.  The sharp variational value is precisely the integrated negative pressure of this layer.

To state the results, set
\[
A(r):=\frac{R_2^3-r^3}{3},
\qquad
q(r):=\frac{r^{7/2}}{\sqrt{r-2M_1}},
\]
and
\[
\Psi(r):=\frac{q(r)}{A(r)}
=
\frac{3r^{7/2}}{\sqrt{r-2M_1}(R_2^3-r^3)}.
\]
There is a unique $x_0\in(2M_1,7M_1/3)$ satisfying
\[
M_1x_0^3+3R_2^3x_0-7M_1R_2^3=0.
\]
Define
\[
a:=\max\{R_1,x_0\},
\qquad
\lambda_*:=\Psi(a),
\qquad
H_*:=\frac1{A(a)}=\frac{3}{R_2^3-a^3},
\]
and
\[
u_*(r):=H_*\one_{(a,R_2)}(r).
\]
{See Figure~\ref{fig:selected-limit} for the graph of  $u_*$. The interior phase requires $R_1\leq x_0$, while the standing assumption $R_1>2M_2$ and the estimate $x_0<7M_1/3$ imply
\[
2M_2<R_1\leq x_0<\frac{7M_1}{3}.
\]
 Outside this parameter range, the optimizer is necessarily pinned to the inner boundary, and only the linear stability regime occurs.}

\begin{theorem}[Overview of the main results]\label{thm:main}
Under the assumptions above, the following statements hold.

\begin{enumerate}[label=\textup{(\roman*)}]

\item The exact $L^1$-relaxation of the smooth functional is $\JBV$.  Its minimum and the smooth infimum agree:
\[
\min_{u\in\ABV}\JBV[u]
=
\inf_{u\in\Ainf}\I[u]
=
\lambda_*.
\]
The unique relaxed minimizer is $u_*$, and no smooth profile attains the infimum.

\item If $u_n\in\Ainf$ and $\I[u_n]\longrightarrow\lambda_*$, then
\[
u_n\longrightarrow u_*
\quad\text{in }L^1(\R),
\]
\[
(u_n')_+\,dr\weakstar H_*\delta_a,
\qquad
(u_n')_-\,dr\weakstar H_*\delta_{R_2},
\]
and the convergence is strict in $BV$.

\item Let $u\in\ABV$, define $\delta[u]:=\mathcal J[u]-\lambda_\ast$, and assume $\delta[u]\leq1$.  If $a=x_0$, then
\[
\|u-u_*\|_{L^1}
+
\|F_u-F_*\|_{L^\infty}
+
\|F_u^{-1}-F_*^{-1}\|_{L^\infty}
\leq C\sqrt{\delta[u]}.
\]
If $a=R_1>x_0$, the right-hand side improves to $C\delta[u]$.  Both exponents are optimal.

\item As the inner endpoint $R_1$ varies, the sharp value $R_1\mapsto
\inf_{u\in\Ainf(R_1,R_2)}\I[u]$ undergoes a variational phase transition at $R_1=x_0$: the optimal rising radius is pinned at the interior point $x_0$ for $R_1\leq x_0$ and at the boundary $R_1$ for $R_1>x_0$.  If $R_2-R_1=L\downarrow0$ with $R_1$ fixed, then
\[
\lambda_*
=
\frac{R_1^{3/2}}{\sqrt{R_1-2M_1}}\frac1L+O(1).
\]
The interior phase can occur only if $x_0>2M_2$. Since $x_0<7M_1/3$, this requires
\[
\frac{\Delta M}{M_1}<\frac16.
\]
If this condition fails, the optimizer is always pinned to the inner boundary and only the linear stability regime occurs.

\item For every smooth minimizing sequence, the associated metrics converge uniformly to a unique locally Lipschitz metric $g_*$ and
\[
G[g_{u_n}]\longrightarrow G[g_*]
\]
as tensor distributions.  Writing $\beta:=\Delta M H_*$, the nonzero mixed components are
\[
G^t{}_t[g_*]=G^r{}_r[g_*]
=-2\beta\one_{(a,R_2)},
\]
\[
G^\theta{}_\theta[g_*]=G^\phi{}_\phi[g_*]
=-2\beta\one_{(a,R_2)}-\beta a\delta_a+\beta R_2\delta_{R_2}.
\]
No $\delta'$ term occurs.

\item The two Dirac terms are exactly the Israel surface tensors at $r=a$ and $r=R_2$.  Both layers have zero surface density, and their tangential pressures are
\[
p_a=-\frac{\beta a}{8\pi\sqrt{1-2M_1/a}},
\qquad
p_{R_2}=\frac{\beta R_2}{8\pi\sqrt{1-2M_2/R_2}}.
\]
The negative tangential NEC measures converge to a Dirac measure at the inner layer, and
\[
4\pi a^2(-p_a)=\frac{\Delta M}{2}\lambda_*.
\]
\end{enumerate}
\end{theorem}

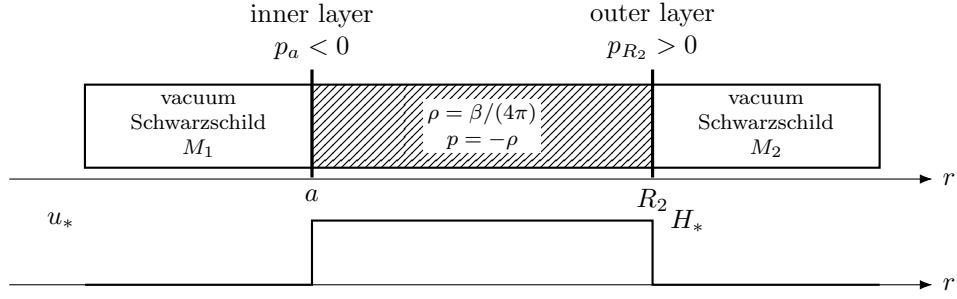
\begin{figure}[t]
\centering
\begin{tikzpicture}[x=1cm,y=1cm,>=Latex,font=\small]
  \draw[->] (0,0) -- (12.2,0) node[right] {$r$};
  \draw[thick] (1,0.15) rectangle (4,1.25);
  \draw[thick,pattern=north east lines,pattern color=black] (4,0.15) rectangle (8.5,1.25);
  \draw[thick] (8.5,0.15) rectangle (11.5,1.25);
  \draw[very thick] (4,0.02) -- (4,1.45);
  \draw[very thick] (8.5,0.02) -- (8.5,1.45);
  \node[align=center,font=\scriptsize] at (2.5,0.7) {vacuum\\Schwarzschild\\$M_1$};
  \node[fill=white,inner sep=1.5pt,align=center,font=\scriptsize] at (6.25,0.7) {$\rho=\beta/(4\pi)$\\$p=-\rho$};
  \node[align=center,font=\scriptsize] at (10,0.7) {vacuum\\Schwarzschild\\$M_2$};
  \node[below] at (4,0) {$a$};
  \node[below] at (8.5,0) {$R_2$};
  \node[above,align=center] at (4,1.45) {inner layer\\$p_a<0$};
  \node[above,align=center] at (8.5,1.45) {outer layer\\$p_{R_2}>0$};
  \draw[->] (0,-1.4) -- (12.2,-1.4) node[right] {$r$};
  \draw[thick] (1,-1.4) -- (4,-1.4) -- (4,-0.55) -- (8.5,-0.55) -- (8.5,-1.4) -- (11.5,-1.4);
  \node[left] at (1,-0.55) {$u_*$};
  \node[right] at (8.6,-0.55) {$H_*$};
\end{tikzpicture}
\caption{{The variationally selected limit.  A normalized box profile produces a constant-density $p=-\rho$ annulus between two vacuum Schwarzschild regions and two timelike surface layers.  All negative tangential NEC concentrates on the inner layer.}}
\label{fig:selected-limit}
\end{figure}

The paper is organized as follows.  Section~\ref{sec:geometry} derives the NEC functional {and an obstruction for the tangential NEC.}  Section~\ref{sec:sharp} proves the exact relaxation theorem and solves the sharp problem.  Section~\ref{sec:stability} establishes the deficit decomposition, sharp stability, and convergence of all minimizing sequences.  Section~\ref{sec:phase} describes the phase transition and thin-annulus asymptotics.  Section~\ref{sec:einstein} proves distributional Einstein-tensor convergence.  Section~\ref{sec:israel} identifies the Israel layers and the sharp physical cost.  The final section records the scope and natural extensions of the model.

\section{Geometric set-up and the NEC functional}\label{sec:geometry}

\begin{definition}[Admissible smooth profiles]
Let $\Ainf$ be the set of functions $u\in C^\infty(\R)$ satisfying {
\[u\geq0, \qquad \supp u\subset(R_1,R_2),
\qquad \int_{R_1}^{R_2}r^2u(r)\,dr=1.\]}
For $u\in\Ainf$, extend $P_u$ by $0$ on $(-\infty,R_1]$ and by $1$ on $[R_2,\infty)$.
\end{definition}

The condition $R_1>2M_2$ gives the uniform bound
\[
F_u(r)
\geq
\kappa:=1-\frac{2M_2}{R_1}>0
\]
for every admissible $u$ and $r\in[R_1,R_2]$.

\begin{proposition}[Einstein tensor and matter variables]\label{prop:einstein-smooth}
Let
\[
g=-F(r)\,dt^2+F(r)^{-1}\,dr^2+r^2d\Omega^2,
\qquad
F(r)=1-\frac{2m(r)}r,
\]
with $m$ smooth.  Then
\[
G^t{}_t=G^r{}_r=-\frac{2m'}{r^2},
\qquad
G^\theta{}_\theta=G^\phi{}_\phi=-\frac{m''}{r}.
\]
If $G^\mu{}_\nu=8\pi\diag(-\rho,p_r,p_t,p_t)$, then
\[
\rho=\frac{m'}{4\pi r^2},
\qquad
p_r=-\rho,
\qquad
p_t=-\frac{m''}{8\pi r}.
\]
For $m=m_u$ one has
\[
\rho=\frac{\Delta M}{4\pi}u,
\qquad
\rho+p_r=0,
\qquad
\rho+p_t=-\frac{\Delta M}{8\pi}ru'.
\]
\end{proposition}

\begin{proof}
For a smooth $F$, direct computation gives
\[
G^t{}_t=G^r{}_r=\frac{rF'+F-1}{r^2},
\qquad
G^\theta{}_\theta=G^\phi{}_\phi=\frac12F''+\frac{F'}r.
\]
Substituting $F=1-2m/r$ yields the first displayed formulas.  Since
\[
m_u'=\Delta M r^2u,
\qquad
m_u''=\Delta M(2ru+r^2u'),
\]
the remaining identities follow.
\end{proof}

{The following elementary obstruction motivates the variational problem; related
analyses of energy conditions in smooth spherical transition geometries appear in
\cite{ConboyLake2005,Zaslavskii2010}. The lower bound will be sharpened in
Theorem~\ref{thm:sharp}.
For a real number $x$, we write
\[
x_+:=\max\{x,0\},
\qquad
x_-:=\max\{-x,0\}.
\]
Thus $(u')_+$ denotes the pointwise positive part of the ordinary derivative
of a smooth function.
\begin{proposition}[Elementary NEC obstruction]\label{thm:nogo}
{For every $u\in\Ainf$, the inequality $\rho+p_t\geq0$ fails on a set of positive Lebesgue measure in $(R_1,R_2)$.  Moreover,}
\[
{\I[u]\geq \frac{R_1^3}{R_2^2(R_2-R_1)}>0.}
\]
{Thus, for fixed $\Delta M$, no admissible smooth Schwarzschild-mass interpolation in this gauge can make the total tangential NEC defect arbitrarily small.}
\end{proposition}

\begin{proof}
{Proposition~\ref{prop:einstein-smooth} gives $\rho+p_t=-(\Delta M/8\pi)ru'$.  Since $u$ is nonnegative, vanishes near both endpoints, and has positive normalization, it is not identically zero.  If $u_{\max}:=\max u$, then $u_{\max}>0$, and the rise from $0$ to $u_{\max}$ forces $u'>0$ on a set of positive measure.  Hence $\rho+p_t<0$ there.}

{For the lower bound, the total upward variation satisfies
\[
\int_{R_1}^{R_2}(u')_+\,dr\geq u_{\max}.
\]
Since $r\geq R_1$ and $F_u\leq1$,
\[
\I[u]\geq R_1^3\int_{R_1}^{R_2}(u')_+\,dr\geq R_1^3u_{\max}.
\]
The normalization gives
\[
1=\int_{R_1}^{R_2}r^2u(r)\,dr
\leq R_2^2(R_2-R_1)u_{\max}.
\]
Combining the last two estimates proves the claim.}
\end{proof}}

On a static slice, the proper spatial volume element is
\[
dV=\frac{4\pi r^2}{\sqrt{F_u(r)}}\,dr.
\]
Therefore the integrated negative tangential NEC is
\[
\int_{R_1}^{R_2}[-(\rho+p_t)]_+\,dV
=
\frac{\Delta M}{2}\I[u],
\]
where
\[
\I[u]
:=
\int_{R_1}^{R_2}
\frac{r^3(u')_+}{\sqrt{F_u}}\,dr.
\]
Thus $\I$ is the physical cost up to the fixed factor $\Delta M/2$.  The ansatz forces $p_r=-\rho$; consequently, this problem concerns the negative tangential NEC defect within a restricted Schwarzschild-gauge class, not a general measure of all stress-energy magnitudes.

\section{Exact relaxation and sharp minimization}\label{sec:sharp}

{Minimizing sequences may converge to profiles with jumps, so the smooth functional must be completed in a topology compatible with compactness.  This section identifies that completion exactly, despite the nonlocal dependence of the weight on the profile.}

We work with the one-dimensional space of functions of bounded variation,  $BV(\R)$.  A function $u\in L^1(\R)$ belongs to $BV(\R)$ when its distributional derivative $Du$ is a finite signed Radon measure.  The Jordan decomposition is
\[
Du=(Du)^+-(Du)^-,
\]
where $(Du)^+$ and $(Du)^-$ are mutually singular positive measures.  Here $(Du)^+$ and $(Du)^-$ are the positive and negative measures in the
Jordan decomposition of the signed measure $Du$. They should not be confused
with the pointwise scalar operations $x\mapsto x_+$ and $x\mapsto x_-$.
For a smooth function $u$, the two notions are related by
\[
d(Du)^+=(u')_+\,dr,
\qquad
d(Du)^-=(u')_-\,dr.
\]
Standard compactness and lower-semicontinuity results may be found in \cite{AmbrosioFuscoPallara2000,EvansGariepy2015}.

\begin{definition}[$BV$ admissible class]
Let $\ABV$ consist of all $u\in BV(\R)$ such that
\[
u\geq0\text{ a.e.},
\qquad
u=0\text{ a.e. on }\R\setminus[R_1,R_2],
\qquad
\int_{R_1}^{R_2}r^2u(r)\,dr=1.
\]

Notice that we use a.e. because an element of $BV(\mathbb{R})$ is usually regarded as an equivalence class of functions modulo equality almost everywhere.
{For $u\in\ABV$, one can define
$P_u, m_u, F_u, g_u $ as for the smooth $u$.
We extend $P_u$ by $0$ on
$(-\infty,R_1]$ and by $1$ on $[R_2,\infty)$, and correspondingly extend
$m_u$ by $M_1$ and $M_2$ on these regions.} Define
\[
\JBV[u]
:=
\int_{[R_1,R_2]}W_u(r)\,d(Du)^+(r),
\qquad
W_u(r):=\frac{r^3}{\sqrt{F_u(r)}}.
\]
For $u\in\Ainf$, one has $d(Du)^+=(u')_+\,dr$, and hence $\JBV[u]=\I[u]$.
\end{definition}

\begin{definition}[$L^1$-relaxation]\label{def:relaxation}
For $u\in\ABV$, define
\[
\I_{\mathrm{rel}}[u]
:=
\inf
\left\{
\liminf_{n\longrightarrow\infty}\I[u_n]:
 u_n\in\Ainf,\ u_n\longrightarrow u\text{ in }L^1(\R)
\right\}.
\]
The value is $+\infty$ when no approximating sequence exists.  {We will show in Theorem~\ref{thm:relax} that $\I_{\mathrm{rel}}[u]$ coincides with $\mathcal J[u]$. Physically, $(\Delta M/2)\I_{\mathrm{rel}}[u]$ is the least proper-volume tangential-NEC defect among smooth matter layers that converge to the idealized profile $u$. A finite value means that the idealized profile can be approximated with uniformly controlled total negative tangential NEC.}
\end{definition}

\begin{lemma}[Continuity of the nonlocal weight]\label{lem:weight-cont}
If $u,v\in\ABV$, then
\[
\|P_u-P_v\|_{L^\infty}
\leq R_2^2\|u-v\|_{L^1}.
\]
Consequently, $u_n\longrightarrow u$ in $L^1$ implies
\[
F_{u_n}\longrightarrow F_u,
\qquad
W_{u_n}\longrightarrow W_u
\]
uniformly on $[R_1,R_2]$.
\end{lemma}

\begin{proof}
For each $r\in[R_1,R_2]$,
\[
|P_u(r)-P_v(r)|
\leq
\int_{R_1}^{R_2}s^2|u(s)-v(s)|\,ds
\leq
R_2^2\|u-v\|_{L^1}.
\]
The formula for $F_u$ now gives uniform convergence of $F_{u_n}$.  Since every admissible profile satisfies
$F_u(r)\geq\kappa>0,
$
the function $x\mapsto x^{-1/2}$ is uniformly Lipschitz on the relevant range.  This proves uniform convergence of $W_{u_n}$.
\end{proof}

{The next lemma is a weighted form of the standard lower semicontinuity of positive variation in $BV$; compare \cite[Chapters 3--5]{AmbrosioFuscoPallara2000} and \cite[Chapter 5]{EvansGariepy2015}.  We give the argument because the weight is nonlocal and varies with the sequence.}

\begin{lemma}[Lower bound under $L^1$ convergence]\label{lem:lsc}
If $u_n\in\Ainf$ and $u_n\longrightarrow u\in\ABV$ in $L^1$, then
\[
\JBV[u]
\leq
\liminf_{n\longrightarrow\infty}\I[u_n].
\]
More generally, the same lower-semicontinuity statement holds for sequences in $\ABV$ with $\JBV[u_n]$ in place of $\I[u_n]$.
\end{lemma}

\begin{proof}
{Set $\mu_n:=(Du_n)^+$.  If $\liminf_n\I[u_n]=+\infty$, the assertion is immediate, so pass to a subsequence along which the liminf is a finite limit.  Because $W_{u_n}(r)\geq R_1^3$,}
\[
R_1^3\mu_n([R_1,R_2])
\leq
\int W_{u_n}\,d\mu_n
=
\I[u_n].
\]
{Hence the positive measures $\mu_n$ have uniformly bounded mass.  By weak-* compactness of bounded Radon measures, after passing to a further subsequence, $\mu_n\weakstar\mu$ for some positive measure $\mu$.}

{We next identify what the limit measure must dominate.  Since $u_n\longrightarrow u$ in $L^1$, the distributional derivatives satisfy $Du_n\longrightarrow Du$ in the sense of distributions: for every $\varphi\in C_c^1(\R)$,}
\[
\int \varphi\,dDu_n
=-\int u_n\varphi'\,dr
\longrightarrow
-\int u\varphi'\,dr
=
\int\varphi\,dDu.
\]
{For every nonnegative $\varphi\in C_c(\R)$, the positive variation has the dual characterization}
\[
\int\varphi\,d(Du)^+
=
\sup\left\{
\int\psi\,dDu:
\psi\in C_c(\R),\ 0\leq\psi\leq\varphi
\right\}.
\]
{Applying this formula first to $Du_n$, passing to the limit for each fixed test function, and then taking the supremum gives}
\[
\int\varphi\,d(Du)^+
\leq
\liminf_{n\longrightarrow\infty}
\int\varphi\,d(Du_n)^+.
\]
{Thus the weak-* limit $\mu$ dominates $(Du)^+$.  This is the familiar lower semicontinuity of positive variation.}

{It remains to handle the varying weight.  Lemma~\ref{lem:weight-cont} gives uniform convergence $W_{u_n}\longrightarrow W_u$.  Since the masses of $\mu_n$ are bounded,}
\[
\left|
\int(W_{u_n}-W_u)\,d\mu_n
\right|
\leq
\|W_{u_n}-W_u\|_{L^\infty}\mu_n([R_1,R_2])
\longrightarrow0.
\]
{Therefore}
\[
\begin{aligned}
\JBV[u]
&=
\int W_u\,d(Du)^+\\
&\leq
\liminf_{n\longrightarrow\infty}
\int W_u\,d(Du_n)^+\\
&=
\liminf_{n\longrightarrow\infty}
\int W_{u_n}\,d(Du_n)^+
=
\liminf_{n\longrightarrow\infty}\I[u_n].
\end{aligned}
\]
{The same reasoning applies to a sequence in $\ABV$, replacing $\I[u_n]$ by $\JBV[u_n]$.}
\end{proof}

The recovery inequality requires a little care because admissible smooth functions must vanish in neighborhoods of both endpoints.  We first move an arbitrary $BV$ profile slightly into the interval and then mollify it.

The two ingredients of the recovery construction---an affine compression away
from the endpoints, followed by mollification---are standard in
$\mathrm{BV}$ relaxation \cite{AmbrosioFuscoPallara2000,EvansGariepy2015}. We record them in a single
lemma with a compressed proof; the point requiring genuine care in this
problem is not the construction but the sequence-dependent weight, handled in
Lemmas~\ref{lem:weight-cont} and~\ref{lem:lsc}.

{
\begin{lemma}[Recovery construction]\label{lem:recovery}
Let $u\in\ABV$. There exist $u_n\in\Ainf$ with
$u_n\to u$ in $L^1(\mathbb R)$ and
\[
  \int W_u\,(u_n')^+\,dr\ \longrightarrow\ \int W_u\,d(Du)^+ .
\]
\end{lemma}

\begin{proof}
\emph{Compression.} Let $c=(R_1+R_2)/2$, $\lambda_\varepsilon=1-4\varepsilon/(R_2-R_1)$
and $T_\varepsilon(r)=c+\lambda_\varepsilon(r-c)$, an increasing affine map of
$[R_1,R_2]$ onto $[R_1+2\varepsilon,R_2-2\varepsilon]$. Put
$v_\varepsilon:=u\circ T_\varepsilon^{-1}$. Then $Dv_\varepsilon=(T_\varepsilon)_\#Du$, and since
$T_\varepsilon$ is increasing it preserves the Jordan decomposition, so
$(Dv_\varepsilon)^\pm=(T_\varepsilon)_\#(Du)^\pm$. As $T_\varepsilon\to\mathrm{id}$
uniformly and $W_u$ is continuous,
\[
  \int W_u\,d(Dv_\varepsilon)^+=\int W_u(T_\varepsilon(r))\,d(Du)^+(r)
  \longrightarrow\int W_u\,d(Du)^+ ,
\]
and a change of variables gives $v_\varepsilon\to u$ in $L^1$ together with
$\int r^2v_\varepsilon\,dr=\lambda_\varepsilon\int T_\varepsilon(s)^2u(s)\,ds\to1$.

\emph{Mollification.} Fix $\varepsilon$ and let $w=v_\varepsilon*\rho_\delta$ with
$\delta<\varepsilon$, so that $w\in C_c^\infty((R_1,R_2))$ is nonnegative and
$w\to v_\varepsilon$ in $L^1$ as $\delta\downarrow0$. Since
$Dw=(Dv_\varepsilon)*\rho_\delta$ and the positive part of a difference of
positive measures is dominated by the first,
$(Dw)^+\le(Dv_\varepsilon)^+*\rho_\delta$; testing against the continuous
weight and using $\check\rho_\delta*W_u\to W_u$ uniformly gives
\[
  \limsup_{\delta\downarrow0}\int W_u\,(w')^+\,dr
  \ \le\ \int W_u\,d(Dv_\varepsilon)^+ .
\]
The reverse inequality is the lower semicontinuity of Lemma~\ref{lem:lsc} applied with
the fixed weight $W_u$, so the limit exists and equals
$\int W_u\,d(Dv_\varepsilon)^+$.

\emph{Diagonal and normalization.} Choose $\varepsilon_n\downarrow0$ and, for
each $n$, a mollification $w_n$ of $v_{\varepsilon_n}$ with
$\|w_n-v_{\varepsilon_n}\|_{L^1}\le1/n$,
$\bigl|\int W_u(w_n')^+dr-\int W_u\,d(Dv_{\varepsilon_n})^+\bigr|\le1/n$ and
$\bigl|\int r^2w_n\,dr-\int r^2v_{\varepsilon_n}\,dr\bigr|\le1/n$. 
Setting $c_n:=\bigl(\int r^2w_n\,dr\bigr)^{-1}$ and $u_n:=c_nw_n$ gives
$u_n\in\Ainf$. Since $c_n\to1$ and $\|w_n\|_{L^1}$ is bounded,
$u_n\to u$ in $L^1$; and since $c_n>0$ gives $(u_n')^+=c_n(w_n')^+$, the
energy limit is unchanged.
\end{proof}
}

\begin{theorem}[Exact $L^1$-relaxation]\label{thm:relax}
For every $u\in\ABV$,
\[
  \I_{\mathrm{rel}}[u]=\mathcal J[u].
\]
In particular the infimum in Definition~\ref{def:relaxation} is over a nonempty set and is
finite, and it is attained along a recovery sequence.
\end{theorem}

{\begin{proof}
Throughout, recall the two-sided bound on the weight: since
$0<\kappa\le F_u\le1$ on $[R_1,R_2]$ for every $u\in\ABV$,
\begin{equation}\label{eq:weightbounds}
  R_1^3\ \le\ W_u(r)=\frac{r^3}{\sqrt{F_u(r)}}\ \le\ \frac{R_2^3}{\sqrt\kappa},
  \qquad r\in[R_1,R_2].
\end{equation}
The upper bound gives $\mathcal J[u]\le\kappa^{-1/2}R_2^3\,(Du)^+([R_1,R_2])<\infty$,
so both sides of the asserted identity are finite.

\smallskip
\emph{Step 1: $\mathcal J[u]\le \I_{\mathrm{rel}}[u]$.}
Let $u_n\in\Ainf$ be any sequence with $u_n\to u$ in $L^1(\mathbb R)$.
Lemma~\ref{lem:lsc} gives
\[
  \mathcal J[u]\ \le\ \liminf_{n\to\infty}\I[u_n].
\]
Taking the infimum over all such sequences yields
$\mathcal J[u]\le \I_{\mathrm{rel}}[u]$. (If no admissible approximating
sequence existed, $\I_{\mathrm{rel}}[u]=+\infty$ and the inequality would be
vacuous; Step 2 shows this does not occur.)

\smallskip
\emph{Step 2: $\I_{\mathrm{rel}}[u]\le\mathcal J[u]$.}
Let $u_n\in\Ainf$ be the sequence supplied by
Lemma~\ref{lem:recovery}, so that $u_n\to u$ in $L^1(\mathbb R)$ and
\begin{equation}\label{eq:fixedweight}
  \int W_u\,(u_n')^+\,dr\ \longrightarrow\ \int W_u\,d(Du)^+=\mathcal J[u].
\end{equation}
We must replace the fixed weight $W_u$ in \eqref{eq:fixedweight} by the
correct sequence-dependent weight $W_{u_n}$.

First, the total positive variations are uniformly bounded: by the lower
bound in \eqref{eq:weightbounds},
\[
  \int (u_n')^+\,dr\ \le\ R_1^{-3}\int W_u\,(u_n')^+\,dr,
\]
and the right-hand side converges by \eqref{eq:fixedweight}, hence
\[
  V:=\sup_n\int (u_n')^+\,dr<\infty.
\]
Second, since $u_n\to u$ in $L^1$, Lemma~\ref{lem:weight-cont} gives
$\|W_{u_n}-W_u\|_{L^\infty([R_1,R_2])}\to0$. Therefore
\[
  \left|\int\bigl(W_{u_n}-W_u\bigr)(u_n')^+\,dr\right|
  \ \le\ \|W_{u_n}-W_u\|_{L^\infty}\int (u_n')^+\,dr
  \ \le\ V\,\|W_{u_n}-W_u\|_{L^\infty}\ \longrightarrow\ 0 .
\]
Combining this with \eqref{eq:fixedweight},
\[
  \I[u_n]=\int W_{u_n}(u_n')^+\,dr
  =\int W_u\,(u_n')^+\,dr+o(1)
  \ \longrightarrow\ \mathcal J[u].
\]
Thus $(u_n)$ is an admissible competitor in Definition~\ref{def:relaxation} with
$\lim_n \I[u_n]=\mathcal J[u]$, whence
$\I_{\mathrm{rel}}[u]\le\mathcal J[u]$.

\smallskip
The two steps give $\I_{\mathrm{rel}}[u]=\mathcal J[u]$, and Step 2 exhibits a
sequence along which the infimum is attained.
\end{proof}

We now solve the relaxed problem.  Define
\[
A(r):=\int_r^{R_2}s^2\,ds=\frac{R_2^3-r^3}{3}.
\]

\begin{lemma}[Exact normalization identity]\label{lem:normalization}
Let $u\in\ABV$ and write
\[
\mu:=(Du)^+,
\qquad
\nu:=(Du)^-.
\]
Then
\[
1
=
\int A\,d\mu-
\int A\,d\nu.
\]
Moreover,
\[
\mu([R_1,R_2])=\nu([R_1,R_2]).
\]
\end{lemma}

\begin{proof}
Since $A'(r)=-r^2$ and the zero extension is used, $BV$ integration by parts gives
\[
\int A\,dDu
=-\int uA'\,dr
=\int r^2u(r)\,dr
=1.
\]
Substituting $Du=\mu-\nu$ proves the first identity.  The zero extension vanishes at both ends of the real line, so $(Du)(\R)=0$, which is exactly the equality of the total positive and negative variations.
\end{proof}
}

Define
\[
q(r):=\frac{r^3}{\sqrt{1-2M_1/r}}
=\frac{r^{7/2}}{\sqrt{r-2M_1}},
\qquad
\Psi(r):=\frac{q(r)}{A(r)}.
\]

\begin{lemma}[Unique optimal radius]\label{lem:radius}
The function $\Psi$ has exactly one critical point $x_0$ in $(2M_1,R_2)$.  It is
the unique real solution of
\[
M_1x_0^3+3R_2^3x_0-7M_1R_2^3=0
\]
and satisfies
\[
2M_1<x_0<\frac{7M_1}{3}.
\]
{The function $\Psi$ is strictly decreasing on $(2M_1,x_0)$ and strictly
increasing on $(x_0,R_2)$. } Consequently, its unique minimizer on $[R_1,R_2)$ is
\[
a=\max\{R_1,x_0\}.
\]
\end{lemma}

\begin{proof}
Logarithmic differentiation yields
\[
\frac{\Psi'(r)}{\Psi(r)}
=
\frac{M_1r^3+3R_2^3r-7M_1R_2^3}
{r(r-2M_1)(R_2^3-r^3)}.
\]
{
On $(2M_1,R_2)$ the three factors $r$, $r-2M_1$ and $R_2^3-r^3$ are positive, as
is $\Psi=q/A$; hence $\Psi'$ has the sign of the numerator
\[
  N(r):=M_1r^3+3R_2^3r-7M_1R_2^3 .
\]
The function $N$ is strictly increasing on all of $\R$, since
$N'(r)=3M_1r^2+3R_2^3>0$, so it has at most one zero.  At the endpoints of the
interval,
\[
  N(2M_1)=M_1\bigl(8M_1^3-R_2^3\bigr)<0,
  \qquad
  N(R_2)=3R_2^3(R_2-2M_1)>0,
\]
the first because $R_2>2M_2>2M_1$ gives $R_2^3>8M_1^3$.  Hence $N$ has exactly
one zero $x_0$, and $x_0\in(2M_1,R_2)$.  Moreover
\[
  N\!\left(\frac{7M_1}{3}\right)=\frac{343M_1^4}{27}>0,
\]
the terms in $R_2^3$ canceling identically, so monotonicity of $N$ on $\R$
gives $x_0<7M_1/3$ whether or not $7M_1/3$ lies below $R_2$.

Since $N<0$ on $(2M_1,x_0)$ and $N>0$ on $(x_0,R_2)$, the function $\Psi$ is
strictly decreasing on the first interval and strictly increasing on the second.
As $[R_1,R_2)\subset(2M_1,R_2)$, its unique minimizer there is $x_0$ when
$R_1\leq x_0$ and $R_1$ when $R_1>x_0$, that is, $a=\max\{R_1,x_0\}$.}
\end{proof}

Set
\[
\lambda_*:=\Psi(a),
\qquad
H_*:=\frac1{A(a)},
\qquad
u_*(r):=H_*\one_{(a,R_2)}(r).
\]
Write $P_*:=P_{u_*}$, $m_*:=m_{u_*}$, and $F_*:=F_{u_*}$.

\begin{theorem}[Sharp value and unique relaxed minimizer]\label{thm:sharp}
For every $u\in\ABV$,
\[
\JBV[u]\geq\lambda_*.
\]
Equality holds exactly for
\[
u=u_*=H_*\one_{(a,R_2)}.
\]
Consequently,
\[
\min_{\ABV}\JBV
=
\inf_{\Ainf}\I
=
\lambda_*,
\]
and the smooth infimum is not attained.
\end{theorem}

\begin{proof}
Let $\mu=(Du)^+$ and $\nu=(Du)^-$.  Since $P_u\geq0$,
\[
F_u(r)\leq1-\frac{2M_1}{r},
\]
and hence
\[
W_u(r)\geq q(r).
\]
By the definition of $\lambda_*$,
\[
q(r)\geq\lambda_*A(r).
\]
Using Lemma~\ref{lem:normalization},
\[
\begin{aligned}
\JBV[u]
&\geq\int q\,d\mu
\geq\lambda_*\int A\,d\mu\\
&=\lambda_*\left(1+\int A\,d\nu\right)
\geq\lambda_*.
\end{aligned}
\]

Suppose equality holds.  Every nonnegative loss in this chain must vanish.  First, $q-\lambda_*A$ vanishes $\mu$-almost everywhere.  By Lemma~\ref{lem:radius}, its only zero is $a$, so $\mu=H\delta_a$.  Second,
\[
\int A\,d\nu=0.
\]
Because $A>0$ on $[R_1,R_2)$ and $A(R_2)=0$, the measure $\nu$ is supported at $R_2$, say $\nu=K\delta_{R_2}$.  The normalization identity gives $HA(a)=1$, hence $H=H_*$.  Equality of the total positive and negative variations gives $K=H_*$.  Thus
\[
Du=H_*\delta_a-H_*\delta_{R_2},
\]
and the zero extension implies $u=u_*$ almost everywhere.  Conversely, direct substitution gives $\JBV[u_*]=\lambda_*$.

The exact relaxation theorem provides smooth recovery sequences for $u_*$.  Therefore the smooth infimum is $\lambda_*$.  A smooth minimizer would also minimize $\JBV$ and hence equal the discontinuous profile $u_*$ almost everywhere, which is impossible.
\end{proof}

\section{Deficit, sharp stability, and minimizing sequences}\label{sec:stability}

{The deficit identity below separates the three ways in which a profile can lose optimality and turns that separation into sharp control of both the profile and the geometry.}

Define
\[
g(r):=q(r)-\lambda_*A(r),
\qquad
\delta[u]:=\JBV[u]-\lambda_*.
\]
The deficit is nonnegative and vanishes only at $u_*$.  The next identity explains precisely how energy is lost away from the optimizer.

\begin{proposition}[Exact deficit decomposition]\label{prop:deficit}
For every $u\in\ABV$, with $\mu=(Du)^+$ and $\nu=(Du)^-$,
\[
{
\delta[u]
=
\int(W_u-q)\,d\mu
+
\int g\,d\mu
+
\lambda_*\int A\,d\nu.
}
\]
All three terms are nonnegative.
\end{proposition}

\begin{proof}
Add and subtract $q$ and $\lambda_*A$:
\[
\begin{aligned}
\JBV[u]-\lambda_*
={}&
\int(W_u-q)\,d\mu
+
\int(q-\lambda_*A)\,d\mu\\
&+
\lambda_*\left(\int A\,d\mu-1\right).
\end{aligned}
\]
Lemma~\ref{lem:normalization} turns the final parenthesis into $\int A\,d\nu$.  The proof of Theorem~\ref{thm:sharp} shows that each term is nonnegative.
\end{proof}

The three terms have distinct meanings.  The first penalizes upward variation after the cumulative mass $P_u$ has become positive.  Indeed, with $F_0(r)=1-2M_1/r$, rationalization gives
\[
W_u(r)-q(r)
=
\frac{2\Delta M r^2P_u(r)}
{\sqrt{F_u(r)}\sqrt{F_0(r)}
\left(\sqrt{F_0(r)}+\sqrt{F_u(r)}\right)}.
\]
The second term forces the positive derivative measure toward the optimal radius $a$.  The third forces the negative derivative measure toward $R_2$, because $A$ vanishes only there.

\begin{lemma}[Coercivity of the localization gap]\label{lem:coercivity}
There are two regimes.

\begin{enumerate}[label=\textup{(\alph*)}]
\item If $R_1\leq x_0$, including the critical boundary case $R_1=x_0$, then $a=x_0$ and there exists $c_2>0$ such that
\[
g(r)\geq c_2|r-a|^2
\]
for every $r\in[R_1,R_2]$.

\item If $R_1>x_0$, then $a=R_1$ and there exists $c_1>0$ such that
\[
g(r)\geq c_1(r-R_1)
\]
for every $r\in[R_1,R_2]$.
\end{enumerate}
\end{lemma}

\begin{proof}
Since $g=A(\Psi-\lambda_*)$, the function $g$ is nonnegative and vanishes only at $a$.  In the first regime, $\Psi'(a)=0$.  The critical point is nondegenerate.  To see this, set
\[
L(r):=\frac{\Psi'(r)}{\Psi(r)}.
\]
The numerator of $L$ has a simple zero at $a$ and strictly positive derivative, so $L'(a)>0$.  Since
\[
\frac{\Psi''}{\Psi}=L'+L^2,
\]
we obtain $\Psi''(a)>0$ and hence $g''(a)=A(a)\Psi''(a)>0$.  The quotient $g(r)/|r-a|^2$ therefore extends continuously to a positive value at $a$ and is positive on the rest of the compact interval.

In the strict boundary regime, $\Psi'(R_1)>0$, and
\[
g'(R_1)=A(R_1)\Psi'(R_1)>0.
\]
The quotient $g(r)/(r-R_1)$ extends continuously and positively at $R_1$, which proves the linear bound.
\end{proof}

The following standard representation converts concentration of derivative measures into $L^1$ control of the profile.

\begin{lemma}[Cumulative representation and transport identities]\label{lem:cumulative}
Let $u\in\ABV$, write $Du=\mu-\nu$, and choose the right-continuous representative.  Then, for almost every $r\in(R_1,R_2)$,
\[
u(r)=\mu([R_1,r])-\nu([R_1,r]).
\]
If
\[
V:=\mu([R_1,R_2])=\nu([R_1,R_2]),
\]
then
\[
\int_{R_1}^{R_2}
\left|
\mu([R_1,r])-V\one_{(a,R_2)}(r)
\right|dr
=
\int|s-a|\,d\mu(s),
\]
and
\[
\int_{R_1}^{R_2}\nu([R_1,r])\,dr
=
\int(R_2-s)\,d\nu(s).
\]
\end{lemma}

\begin{proof}
The first formula is the fundamental theorem for one-dimensional $BV$ functions applied to the zero extension.  
The last formula follows as 
\[
\begin{aligned}
\int_{R_1}^{R_2}\nu([R_1,r])\,dr
&=
\int_{R_1}^{R_2}\int_{[R_1,R_2]}
\mathbf{1}_{\{s\le r\}}\,d\nu(s)\,dr\\
&=
\int_{[R_1,R_2]}
\left(\int_s^{R_2}dr\right)d\nu(s)
=
\int_{[R_1,R_2]}(R_2-s)\,d\nu(s).
\end{aligned}
\]
Similarly, for the second formula, split the $r$-integral at $a$.  Before $a$, the cumulative mass $\mu([R_1,r])$ counts atoms lying to the left of $r$; after $a$, the difference $V-\mu([R_1,r])$ counts atoms lying to the right.  Tonelli's theorem then gives exactly the total transport cost $\int|s-a|\,d\mu(s)$.  
\end{proof}

\begin{theorem}[Sharp quantitative stability and metric control]\label{thm:stability}
There is a constant $C$, depending only on $M_1,M_2,R_1,R_2$, such that the following holds whenever $u\in\ABV$ and the deficit satisfies $\delta[u]\leq1$.

If $a=x_0$, then
\[
\int|r-a|^2\,d(Du)^+(r)
\leq C\delta[u],
\]
\[
\int(R_2-r)\,d(Du)^-(r)
\leq C\delta[u],
\]
and
\[
\|u-u_*\|_{L^1}
\leq C\sqrt{\delta[u]}.
\]
Moreover,
\[
\|P_u-P_*\|_{L^\infty}
+
\|m_u-m_*\|_{L^\infty}
+
\|F_u-F_*\|_{L^\infty}
+
\|F_u^{-1}-F_*^{-1}\|_{L^\infty}
\leq C\sqrt{\delta[u]}.
\]

If $a=R_1>x_0$, then
\[
\int(r-R_1)\,d(Du)^+(r)
\leq C\delta[u],
\]
\[
\int(R_2-r)\,d(Du)^-(r)
\leq C\delta[u],
\]
and all the preceding $L^1$ and metric estimates hold with $C\delta[u]$ in place of $C\sqrt{\delta[u]}$.
\end{theorem}

\begin{proof}
Write $\mu=(Du)^+$, $\nu=(Du)^-$, and
\[
V:=\mu([R_1,R_2])=\nu([R_1,R_2]).
\]
Proposition~\ref{prop:deficit} and Lemma~\ref{lem:coercivity} immediately give the stated moment estimate for $\mu$.  Since
\[
A(r)=\int_r^{R_2}s^2\,ds
\geq R_1^2(R_2-r),
\]
the third term of the deficit decomposition gives
\[
\int(R_2-r)\,d\nu(r)
\leq
\frac{\delta[u]}{\lambda_*R_1^2}.
\]

We next control the total jump height $V$.  Since $W_u\geq R_1^3$ and $\JBV[u]\leq\lambda_*+1$,
\[
V\leq\frac{\lambda_*+1}{R_1^3}.
\]
The normalization identity gives
\[
A(a)V-1
=
\int(A(a)-A(r))\,d\mu(r)
+
\int A(r)\,d\nu(r).
\]
Because $|A'(r)|\leq R_2^2$,
\[
A(a)|V-H_*|
\leq
R_2^2\int|r-a|\,d\mu(r)
+
\int A\,d\nu.
\]
In the quadratic regime, Cauchy--Schwarz and the uniform bound for $V$ imply
\[
\int|r-a|\,d\mu(r)
\leq
V^{1/2}
\left(\int|r-a|^2\,d\mu(r)\right)^{1/2}
\leq C\sqrt{\delta[u]}.
\]
Hence $|V-H_*|\leq C\sqrt{\delta[u]}$.  In the strict boundary regime, the first-moment estimate directly gives $|V-H_*|\leq C\delta[u]$.

Lemma~\ref{lem:cumulative} and the triangle inequality now yield
\[
\begin{aligned}
\|u-u_*\|_{L^1}
\leq{}&
\int|r-a|\,d\mu(r)
+(R_2-a)|V-H_*|\\
&+
\int(R_2-r)\,d\nu(r).
\end{aligned}
\]
Substitution gives the two claimed rates.

Finally,
\[
\|P_u-P_*\|_{L^\infty}
\leq R_2^2\|u-u_*\|_{L^1},
\]
\[
\|m_u-m_*\|_{L^\infty}
\leq\Delta M\|P_u-P_*\|_{L^\infty},
\]
and
\[
\|F_u-F_*\|_{L^\infty}
\leq\frac{2}{R_1}\|m_u-m_*\|_{L^\infty}.
\]
Since both $F_u$ and $F_*$ are bounded below by $\kappa$,
\[
\|F_u^{-1}-F_*^{-1}\|_{L^\infty}
\leq\kappa^{-2}\|F_u-F_*\|_{L^\infty}.
\]
This transfers the profile estimates to the spacetime metric coefficients.
\end{proof}

\begin{proposition}[Optimality of the stability exponents]\label{prop:optimal-rate}
The square-root exponent in the critical regime and the linear exponent in the strict boundary regime are optimal.  They cannot be replaced by larger powers of the deficit, either in the $BV$ class or uniformly in the smooth class.
\end{proposition}

\begin{proof}
For $x\in[R_1,R_2)$, define the normalized box profile
\[
u_x(r):=H_x\one_{(x,R_2)}(r),
\qquad
H_x:=\frac1{A(x)}.
\]
Then
\[
\JBV[u_x]=\Psi(x)
\]
and a direct comparison of the two boxes gives asymptotic equivalence
\[
\|u_x-u_*\|_{L^1}\asymp|x-a|
\]
as $x\longrightarrow a$.  If $a=x_0$, then $\Psi'(a)=0$ and $\Psi''(a)>0$, so
\[
\Psi(x)-\Psi(a)\asymp|x-a|^2.
\]
Hence $\|u_x-u_*\|_{L^1}\asymp\sqrt{\delta[u_x]}$.  If $a=R_1>x_0$, then $\Psi'(R_1)>0$, and
\[
\Psi(R_1+h)-\Psi(R_1)\asymp h.
\]
Thus $\|u_{R_1+h}-u_*\|_{L^1}\asymp\delta[u_{R_1+h}]$.  The exact relaxation theorem supplies smooth recovery sequences for these nearby boxes, so a stronger uniform exponent cannot hold in the smooth class either. {In detail, fix $x\neq a$ and set
$\varepsilon:=|x-a|^3$. By Theorem~\ref{thm:relax} there is $v\in\Ainf$ with
$\|v-u_x\|_{L^1}\le\varepsilon$ and $|\I[v]-\mathcal J[u_x]|\le\varepsilon$.
Since $\|u_x-u_*\|_{L^1}\asymp|x-a|$ and
$\mathcal J[u_x]-\lambda_*\asymp|x-a|^2$ in the critical regime
(resp.\ $\asymp|x-a|$ in the strict boundary regime), the error $\varepsilon$
is negligible against both scales, so
$\|v-u_*\|_{L^1}\asymp\|u_x-u_*\|_{L^1}$ and
$\I[v]-\lambda_*\asymp\mathcal J[u_x]-\lambda_*$. Letting $x\to a$ gives smooth
profiles realizing the two rates.}
\end{proof}

\begin{theorem}[Convergence of every smooth minimizing sequence]\label{thm:minseq}
Let $u_n\in\Ainf$ satisfy
\[
\I[u_n]\longrightarrow\lambda_*.
\]
Then
\[
u_n\longrightarrow u_*
\quad\text{in }L^1(\R),
\]
\[
(u_n')_+\,dr\weakstar H_*\delta_a,
\qquad
(u_n')_-\,dr\weakstar H_*\delta_{R_2},
\]
and
\[
|Du_n|(\R)\longrightarrow|Du_*|(\R)=2H_*.
\]
Thus the convergence is strict in $BV$ and the limiting profile is independent of the chosen minimizing sequence.
\end{theorem}

\begin{proof}
{For clarity, we separate the arguments into compactness, identification of the $L^1$ limit, and convergence of the derivative measures.}

{First, the energy bound controls total variation.  Since $W_{u_n}\geq R_1^3$,}
\[
(Du_n)^+(\R)
=\int_{R_1}^{R_2}(u_n')_+\,dr
\leq
\frac{\I[u_n]}{R_1^3}.
\]
{Because each $u_n$ is smooth and vanishes outside $(R_1,R_2)$, its total derivative on $\R$ is zero.  Hence the total upward and downward variations are equal:}
\[
(Du_n)^+(\R)=(Du_n)^-(\R).
\]
{The normalization also gives an $L^1$ bound, since}
\[
\|u_n\|_{L^1}
\leq
\frac1{R_1^2}
\int_{R_1}^{R_2}r^2u_n(r)\,dr
=
\frac1{R_1^2}.
\]
{Thus $(u_n)$ is bounded in $BV([R_1,R_2])$.  The standard compactness theorem for $BV$ functions \cite{AmbrosioFuscoPallara2000,EvansGariepy2015} implies that every subsequence has a further subsequence converging in $L^1$ to some $u\in BV$.  Nonnegativity, support, and the weighted normalization pass to the limit, so $u\in\ABV$.}

{For such a convergent subsequence, Lemma~\ref{lem:lsc} yields}
\[
\JBV[u]
\leq
\liminf_{n\longrightarrow\infty}\I[u_n]
=
\lambda_*.
\]
{Since $\lambda_*$ is the minimum and Theorem~\ref{thm:sharp} gives a unique minimizer, $u=u_*$.  Every subsequence therefore has a further subsequence converging to $u_*$.  A standard contradiction argument now shows that the whole sequence converges: if not, some subsequence would remain a fixed positive $L^1$ distance from $u_*$, but that subsequence would itself have a further subsequence converging to $u_*$.  Hence}
\[
u_n\longrightarrow u_*
\quad\text{in }L^1(\R).
\]

{We next identify the positive and negative derivative measures.  Put}
\[
\mu_n=(u_n')_+\,dr,
\qquad
\nu_n=(u_n')_-\,dr.
\]
{The deficit decomposition consists of nonnegative terms.  In particular, for every fixed $\eta>0$,}
\[
\mu_n(\{|r-a|\geq\eta\})
\leq
\frac{\delta[u_n]}{\min_{|r-a|\geq\eta}g(r)}
\longrightarrow0,
\]
{because $g$ is continuous and vanishes only at $a$.  Likewise,}
\[
\nu_n([R_1,R_2-\eta])
\leq
\frac{\delta[u_n]}{\lambda_*A(R_2-\eta)}
\longrightarrow0.
\]
{Thus every weak-* cluster point of $\mu_n$ is supported at $a$, and every cluster point of $\nu_n$ is supported at $R_2$.  It remains to determine their masses.  The normalization identity gives}
\[
\int A\,d\mu_n
=
1+\int A\,d\nu_n.
\]
{The last integral tends to zero.  Indeed, split it at $R_2-\eta$: the first part is bounded by $A(R_1)\nu_n([R_1,R_2-\eta])$, while on the remaining interval $A\leq A(R_2-\eta)$; first let $n\longrightarrow\infty$ and then $\eta\downarrow0$.  Hence $\int A\,d\mu_n\longrightarrow1$.  Since $\mu_n$ concentrates at $a$ and $A$ is continuous,}
\[
A(a)\mu_n(\R)\longrightarrow1,
\qquad
\mu_n(\R)\longrightarrow H_*.
\]
{The equality of upward and downward total variations gives $\nu_n(\R)\longrightarrow H_*$.  Therefore}
\[
\mu_n\weakstar H_*\delta_a,
\qquad
\nu_n\weakstar H_*\delta_{R_2}.
\]
{Finally,}
\[
|Du_n|(\R)
=
\mu_n(\R)+\nu_n(\R)
\longrightarrow2H_*
=
|Du_*|(\R).
\]
{Together with $L^1$ convergence, this is precisely strict convergence in $BV$.}
\end{proof}

\section{Variational phase transition and thin-annulus asymptotics}\label{sec:phase}

{The asymptotic expansions below are elementary consequences of the explicit optimizer.  Their role is interpretive: they explain where the optimal transition layer forms when its location is unconstrained and how the NEC cost blows up when a fixed mass increase is compressed into a thin annulus.}

The explicit optimizer permits a concise description of how the sharp problem changes when the annulus parameters vary.  This is useful physically because it distinguishes a freely selected interior layer from a layer forced against the inner boundary.

\begin{corollary}
    [Interior and boundary phases]\label{thm:phase}
Fix $M_1,M_2,R_2$ and allow the inner endpoint $R_1$ to vary in $(2M_2,R_2)$.  We write $\Ainf(R_1,R_2)$ for the smooth admissible class with these endpoints.  Let
\[
\Lambda(R_1)
:=
\inf_{u\in\Ainf(R_1,R_2)}\I[u].
\]
Then
\[
{
\Lambda(R_1)=\Psi(\max\{R_1,x_0\}).
}
\]
If $x_0>2M_2$, this gives two phases:
\[
\Lambda(R_1)
=
\begin{cases}
\Psi(x_0),&2M_2<R_1\leq x_0,\\
\Psi(R_1),&x_0\leq R_1<R_2.
\end{cases}
\]
Thus the optimal rising radius is the fixed interior point $x_0$ in the first phase and is pinned to $R_1$ in the second.  The function $\Lambda$ is $C^1$ at $R_1=x_0$ but not $C^2$ there: see Figure~\ref{fig:phase}.  The same point separates the square-root stability regime from the linear stability regime.  If $x_0\leq2M_2$, the admissible range contains only the boundary phase.
\end{corollary}

\begin{proof}
For every fixed $R_1$, Theorem~\ref{thm:sharp} gives
\[
\Lambda(R_1)=\min_{r\in[R_1,R_2)}\Psi(r).
\]
Lemma~\ref{lem:radius} shows that $\Psi$ decreases up to $x_0$ and increases after $x_0$, proving the formula and the two phases.  At the transition, the left derivative of $\Lambda$ is zero because the left branch is constant.  The right derivative is $\Psi'(x_0)=0$, so $\Lambda$ is $C^1$.  The left second derivative is zero, whereas the right second derivative is $\Psi''(x_0)>0$, so $\Lambda$ is not $C^2$.  The stability statement follows from Theorem~\ref{thm:stability}.
\end{proof}

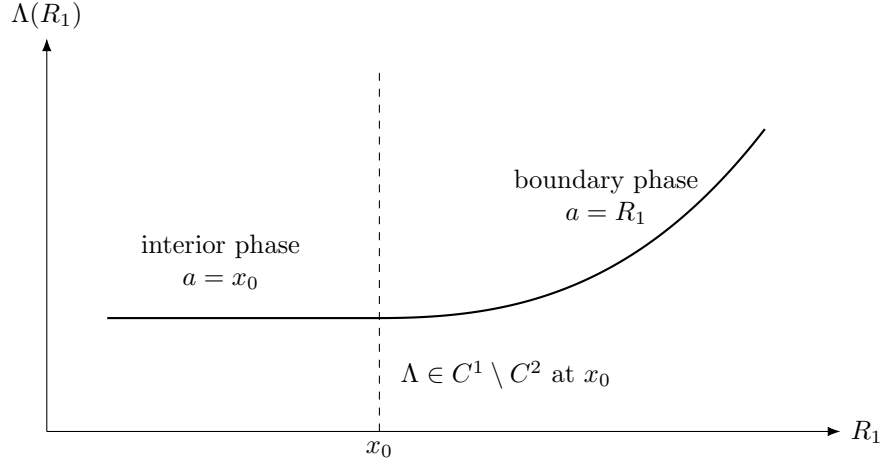
\begin{figure}[t]
\centering
\begin{tikzpicture}[x=1cm,y=1cm,font=\small,>=Latex]
  \draw[->] (0,0) -- (10.5,0) node[right] {$R_1$};
  \draw[->] (0,0) -- (0,5.2) node[above] {$\Lambda(R_1)$};
  \draw[thick] (0.8,1.5) -- (4.4,1.5);
  \draw[thick,domain=4.4:9.5,samples=80] plot (\x,{1.5+0.035*(\x-4.4)^2+0.012*(\x-4.4)^3});
  \draw[dashed] (4.4,0) -- (4.4,4.8);
  \node[below] at (4.4,0) {$x_0$};
  \node[align=center] at (2.3,2.25) {interior phase\\$a=x_0$};
  \node[align=center] at (7.4,3.1) {boundary phase\\$a=R_1$};
  \node[anchor=west] at (4.55,0.75) {$\Lambda\in C^1\setminus C^2$ at $x_0$};
\end{tikzpicture}
\caption{{Qualitative phase diagram for the sharp value $\Lambda(R_1)$.  The cost is constant while the optimal rising radius is the free interior point $x_0$, and increases once the constraint pins the rise to $R_1$.}}
\label{fig:phase}
\end{figure}

{
Figure~\ref{fig:phase} shows the two phases of $\Lambda$. 
}

\begin{theorem}[Thin-annulus divergence]\label{thm:thin}
Fix $M_1<M_2$ and a radius $R>2M_2$.  Set
\[
R_1=R,
\qquad
R_2=R+L,
\]
and let $L\downarrow0$.  For all sufficiently small $L$, the optimizer is pinned at $a=R$.  Moreover,
\[
H_*
=
\frac1{R^2L}+O(1),
\]
\[
\lambda_*
=
\frac{R^{3/2}}{\sqrt{R-2M_1}}\frac1L+O(1),
\]
and the physical minimum satisfies
\[
\frac{\Delta M}{2}\lambda_*
=
\frac{\Delta M R^{3/2}}{2\sqrt{R-2M_1}}\frac1L+O(1).
\]
In particular, a fixed increase of Schwarzschild mass cannot be compressed into a vanishingly thin annulus within this ansatz without an unbounded tangential NEC cost.
\end{theorem}

\begin{proof}
At $R_2=R_1=R$, the numerator controlling the sign of $\Psi'(R)$ is
\[
M_1R^3+3R^4-7M_1R^3
=3R^3(R-2M_1)>0.
\]
By continuity, $\Psi'(R)>0$ for all sufficiently small positive $L$, so the constrained optimizer is $a=R$.  Since
\[
(R+L)^3-R^3
=3R^2L+3RL^2+L^3,
\]
we obtain
\[
H_*
=
\frac{3}{(R+L)^3-R^3}
=
\frac1{R^2L}+O(1).
\]
Substituting $a=R$ into $\lambda_*=\Psi(a)$ gives
\[
\lambda_*
=
\frac{3R^{7/2}}
{\sqrt{R-2M_1}\bigl((R+L)^3-R^3\bigr)}
=
\frac{R^{3/2}}{\sqrt{R-2M_1}}\frac1L+O(1).
\]
Multiplication by $\Delta M/2$ proves the final formula.
\end{proof}

\section{The selected metric and distributional Einstein tensor}\label{sec:einstein}

{We now pass from variational convergence to geometric convergence and compute the distributional Einstein tensor of the uniquely selected metric. The general fact that continuous metrics with a jump in first derivatives may carry hypersurface-supported curvature is classical \cite{GerochTraschen1987,MarsSenovilla1993,LeFlochMardare2007}. In the present paper the new point is not this general mechanism, but that variational near-minimality forces precisely one such metric and allows the distributional limit to be computed directly from every smooth minimizing sequence.}

The limiting cumulative profile is
\[
P_*(r)
=
\begin{cases}
0,&r\leq a,\\
\displaystyle\frac{H_*}{3}(r^3-a^3),&a<r<R_2,\\
1,&r\geq R_2.
\end{cases}
\]
Define
\[
m_*(r):=M_1+\Delta M P_*(r),
\qquad
F_*(r):=1-\frac{2m_*(r)}r,
\]
and
\[
g_*=-F_*dt^2+F_*^{-1}dr^2+r^2d\Omega^2.
\]

\begin{corollary}[Uniform metric convergence]\label{cor:metric}
For every smooth minimizing sequence,
\[
P_{u_n}\longrightarrow P_*,
\qquad
m_{u_n}\longrightarrow m_*,
\qquad
F_{u_n}\longrightarrow F_*,
\qquad
F_{u_n}^{-1}\longrightarrow F_*^{-1}
\]
uniformly on $[R_1,R_2]$.  Moreover, $F_*$ is continuous and locally Lipschitz, so $g_*$ is a continuous locally Lipschitz Lorentzian metric.
\end{corollary}

\begin{proof}
Lemma~\ref{lem:weight-cont} and Theorem~\ref{thm:minseq} give the first three convergences.  Since all coefficients are bounded below by $\kappa$, inversion preserves uniform convergence.  The function $m_*$ is continuous and piecewise $C^1$ with bounded one-sided derivatives, hence $F_*$ is locally Lipschitz.
\end{proof}

In the present symmetry class, the distributional calculation is elementary because the Einstein tensor is linear in $m'$ and $m''$.

Set
\[
\beta:=\Delta M H_*.
\]
Then
\[
m_*'(r)=\beta r^2\one_{(a,R_2)}(r)
\]
as a locally bounded function.  Differentiating as a distribution gives
\[
m_*''
=
2\beta r\one_{(a,R_2)}
+
\beta a^2\delta_a
-
\beta R_2^2\delta_{R_2}.
\]

\begin{definition}[Tensor-distribution convergence in the fixed chart]
All metrics considered here are diagonal in the same coordinates and depend only on $r$.  We say that their mixed Einstein tensors converge as tensor distributions if each mixed component converges in $\mathcal D'$ on every compact coordinate cylinder.  This componentwise definition is coordinate-consistent within the fixed smooth chart.  It is especially convenient here because
\[
\sqrt{-\det g_u}=r^2\sin\theta
\]
is independent of $u$.
\end{definition}

{
Concretely, for a mixed component $T(r)$ and $\varphi\in C_c^\infty(\mathcal U)$
on a coordinate cylinder $\mathcal U$ we use the coordinate pairing
\[
\langle T,\varphi\rangle:=\int T(r)\,\varphi(t,r,\theta,\phi)\,
  dt\,dr\,d\theta\,d\phi .
\]
Because every component depends on $r$ alone, setting
$\tilde\varphi(r):=\int\varphi\,dt\,d\theta\,d\phi\in C_c^\infty$ reduces this
to a one-dimensional pairing $\int T(r)\tilde\varphi(r)\,dr$, and it suffices
to test against $\tilde\varphi$. Note that the reference measure here is the
coordinate $dr$ rather than the invariant volume $dV_g$; the two
normalizations are reconciled in Corollary~\ref{cor:Israel-match}.

\begin{remark}[The Einstein tensor of $g_*$]\label{rem:einstein-of-gstar}
Since $F_*$ is locally Lipschitz and bounded below by $\kappa>0$, the
Christoffel symbols of $g_*$ lie in $L^\infty_{\mathrm{loc}}$; the curvature
is obtained by differentiating them once distributionally, and the
quadratic terms are products of $L^\infty_{\mathrm{loc}}$ functions. No
product of distributions arises, and $g_*$ is a regular metric in the sense
of Geroch--Traschen \cite{GerochTraschen1987}. Moreover, the gauge
$g_{tt}g_{rr}=-1$ makes the mixed components
\[
G^t{}_t=G^r{}_r=\frac{rF'+F-1}{r^2},\qquad
G^\theta{}_\theta=G^\phi{}_\phi=\tfrac12 F''+\frac{F'}{r}
\]
\emph{linear} in $F$, hence in $m$. The formulas of Proposition~\ref{prop:einstein-smooth},
derived there for smooth $m$, therefore extend verbatim to
$m\in W^{1,\infty}_{\mathrm{loc}}$ with all derivatives read in $\mathcal D'$,
and this is the sense in which $G[g_*]$ is understood below.
\end{remark}}

\begin{theorem}[Distributional Einstein convergence]\label{thm:Einstein-convergence}
Let $u_n\in\Ainf$ be any minimizing sequence.  Then
\[
G[g_{u_n}]\longrightarrow G[g_*]
\]
as tensor distributions.  The nonzero mixed components of the limit are
\[
G^t{}_t[g_*]=G^r{}_r[g_*]
=-2\beta\one_{(a,R_2)},
\]
and
\[
G^\theta{}_\theta[g_*]
=G^\phi{}_\phi[g_*]
=-2\beta\one_{(a,R_2)}
-\beta a\delta_a
+\beta R_2\delta_{R_2}.
\]
No derivative of a Dirac measure occurs.
\end{theorem}

{
\begin{proof}
Choose $\varepsilon>0$ sufficiently small that $R_1-\varepsilon>0$, and set
\[
I_\varepsilon:=(R_1-\varepsilon,R_2+\varepsilon).
\]
For every $n$, extend the mass function to $I_\varepsilon$ by
\[
m_{u_n}(r)=
\begin{cases}
M_1, & r\leq R_1,\\
M_1+\Delta M\displaystyle\int_{R_1}^{r}s^2u_n(s)\,ds,
   & R_1<r<R_2,\\
M_2, & r\geq R_2.
\end{cases}
\]
Because $u_n\in C_c^\infty((R_1,R_2))$, this extension is smooth on
$I_\varepsilon$. Similarly, extend $m_*$ by the constants $M_1$ and
$M_2$ outside $[R_1,R_2]$. We test all radial distributions against
functions $\phi\in C_c^\infty(I_\varepsilon)$. In particular, such test
functions may be nonzero at $R_1$ and $R_2$, and hence detect possible
Dirac masses supported at the shell radii.

For the $t$ and $r$ components, Proposition~\ref{prop:einstein-smooth} gives
\[
G^t{}_t[g_{u_n}]
=
G^r{}_r[g_{u_n}]
=
-\frac{2m_{u_n}'}{r^2}.
\]
On $I_\varepsilon$ we have
\[
m_{u_n}'=\Delta M r^2u_n,
\]
where $u_n$ is understood to be extended by zero outside
$(R_1,R_2)$. Since $u_n\to u_*$ in $L^1(\mathbb R)$,
\[
m_{u_n}'\longrightarrow m_*'
   =\beta r^2\mathbf 1_{(a,R_2)}
\qquad\text{in }L^1(I_\varepsilon).
\]
Therefore
\[
G^t{}_t[g_{u_n}]
=
G^r{}_r[g_{u_n}]
\longrightarrow
-2\beta\mathbf 1_{(a,R_2)}
\]
in $L^1(I_\varepsilon)$, and hence in
$\mathcal D'(I_\varepsilon)$.

For the angular components,
\[
G^\theta{}_\theta[g_{u_n}]
=
G^\phi{}_\phi[g_{u_n}]
=
-\frac{m_{u_n}''}{r}.
\]
Let $\phi\in C_c^\infty(I_\varepsilon)$. Since
$m_{u_n}'\to m_*'$ in $L^1(I_\varepsilon)$,
\[
\left\langle m_{u_n}'',\phi\right\rangle
=
-\int_{I_\varepsilon}m_{u_n}'(r)\phi'(r)\,dr
\longrightarrow
-\int_{I_\varepsilon}m_*'(r)\phi'(r)\,dr
=
\left\langle m_*'',\phi\right\rangle.
\]
Thus
\[
m_{u_n}''\longrightarrow m_*''
\qquad\text{in }\mathcal D'(I_\varepsilon).
\]
Because multiplication by the smooth function $1/r$ is continuous on
$\mathcal D'(I_\varepsilon)$, it follows that
\[
-\frac{m_{u_n}''}{r}
\longrightarrow
-\frac{m_*''}{r}
\qquad\text{in }\mathcal D'(I_\varepsilon).
\]

Now
\[
m_*'(r)=\beta r^2\mathbf 1_{(a,R_2)}(r).
\]
Taking its distributional derivative on $I_\varepsilon$ gives
\[
m_*''
=
2\beta r\mathbf 1_{(a,R_2)}
+\beta a^2\delta_a
-\beta R_2^2\delta_{R_2}.
\]
This formula remains valid when $a=R_1$, because the test functions are
defined on the larger interval $I_\varepsilon$ and may be nonzero at
$r=R_1$. Consequently,
\[
G^\theta{}_\theta[g_*]
=
G^\phi{}_\phi[g_*]
=
-2\beta\mathbf 1_{(a,R_2)}
-\beta a\delta_a
+\beta R_2\delta_{R_2}.
\]

Hence, on every coordinate cylinder whose radial projection is contained
in $I_\varepsilon$,
\[
G[g_{u_n}]\longrightarrow G[g_*]
\]
componentwise as tensor distributions, with
\[
G^t{}_t[g_*]
=
G^r{}_r[g_*]
=
-2\beta\mathbf 1_{(a,R_2)}
\]
and
\[
G^\theta{}_\theta[g_*]
=
G^\phi{}_\phi[g_*]
=
-2\beta\mathbf 1_{(a,R_2)}
-\beta a\delta_a
+\beta R_2\delta_{R_2}.
\]
No derivative of a Dirac measure occurs, because $m_*$ is continuous.
A $\delta'$ term could arise only if $m_*$ itself had a jump.
\end{proof}

{
\begin{remark}[Why no $\delta'$: the energy bound forbids it]\label{rem:nodeltaprime}
The absence of a $\delta'$ is not a smoothness accident but a consequence of
the cost functional. Since $u_n\to u_*$ only in $L^1$ while
$G^\theta{}_\theta$ involves $m''$, and bounded sets in $L^1$ are not weakly
compact, nothing a priori prevents $u_n$ from concentrating into a Dirac
mass; were that to happen, $m_*$ would jump and a $\delta'$ would appear.
It is excluded by the lower weight bound $W_u\ge R_1^3$ of
\eqref{eq:weightbounds}: finiteness of
$\mathcal J[u]=\int W_u\,d(Du)^+$ bounds the positive variation of $u$, so
$u_*\in\ABV$ is a function and not a measure. The singularity is thereby
pushed down exactly one derivative: $u_*$ may jump and $m_*'$ may jump, but
$m_*$ remains Lipschitz. This is what keeps $g_*$ Geroch--Traschen regular
(Remark~\ref{rem:einstein-of-gstar}) and what makes the singular part a
genuine Israel single layer rather than a dipole layer, for which no
surface stress tensor exists (Corollary~\ref{cor:Israel-match} below).
\end{remark}}

\begin{remark}
The same conclusion holds for the covariant components. For instance
$G_{tt}[g_{u_n}]=-F_{u_n}G^t{}_t[g_{u_n}]$, and $F_{u_n}\to F_*$ uniformly
with $F_*$ continuous by Corollary~\ref{cor:metric}; multiplication by a uniformly
convergent sequence of continuous functions is compatible with weak-$*$
convergence of the associated measures, so $G_{tt}[g_{u_n}]\to G_{tt}[g_*]$,
including at the two atoms. The statement is thus not an artifact of the
index position.
\end{remark}
}

The bulk tensor on $(a,R_2)$ is
\[
G^\mu{}_\nu=-2\beta\delta^\mu{}_\nu,
\]
corresponding to a positive constant density
\[
\rho=\frac{\beta}{4\pi}
\]
and isotropic pressure $p=-\rho$.  The exterior regions are vacuum Schwarzschild regions of masses $M_1$ and $M_2$. {
\begin{remark}[Mass consistency]\label{rem:mass-consistency}
The bulk accounts for the entire mass increase. Indeed, with the coordinate
element $r^2\,dr$, which is the one defining the Schwarzschild mass function
$m_*$ rather than the proper volume element of Section~2,
\[
  \int_a^{R_2} 4\pi r^2 \rho\, dr
  = \beta \int_a^{R_2} r^2\,dr
  = \beta A(a)
  = \Delta M H_* A(a)
  = \Delta M ,
\]
since $H_* = 1/A(a)$. By Theorem~\ref{thm:israel} below both layers have zero surface
energy density, so neither contributes. Thus all of $\Delta M$ resides in the
constant-density annulus, whereas by Theorem~\ref{thm:concentration} below all of the
negative tangential NEC resides on the inner layer at $r=a$: the selected
geometry separates where the mass sits from where the energy-condition cost is
paid.
\end{remark}

\begin{remark}[Relation to gravastar models]\label{rem:gravastar}
The selected geometry shares two defining features with gravastar models
\cite{MazurMottola2004,VisserWiltshire2004}: a constant-density bulk with $p=-\rho$, and
thin timelike layers matching it to vacuum Schwarzschild exteriors. It differs in
three respects. First, the bulk is not a de Sitter core but a Kottler region: on
$(a,R_2)$ one has $m_*(r) = M_1 + \tfrac{\beta}{3}(r^3-a^3)$, so
\[
  F_*(r) = 1 - \frac{2\bigl(M_1 - \beta a^3/3\bigr)}{r} - \frac{2\beta}{3}\,r^2 ,
\]
a Schwarzschild--de Sitter metric with $\Lambda_{\rm eff} = 2\beta$ and mass parameter
$M_1 - \beta a^3/3$, consistent with $\rho = \beta/(4\pi) = \Lambda_{\rm eff}/(8\pi)$.
Second, the region $r<a$ is vacuum Schwarzschild of mass $M_1>0$ rather than a
regular center, so the construction is a shell interpolation around a pre-existing
mass and not a horizonless compact object of the kind gravastars are proposed to
be. Third, both layers carry zero surface density
(Theorem~\ref{thm:israel}), whereas gravastar shells generically do not. We use
``gravastar-type'' below only for the bulk--two-layer structure.
\end{remark}}

\section{Israel surface layers and the sharp cost}\label{sec:israel}

{Israel's junction formula and its equivalence with the distributional description are established results \cite{Israel1966,Poisson2004,MansouriKhorrami1996}.  We repeat the short calculation to fix signs and normalization. The new conclusion is the variational one: the selected inner surface pressure is exactly the sharp limiting NEC cost and is reached by every smooth minimizing sequence.
}

 We use the normal pointing toward increasing $r$ and the convention
\[
K_{AB}=e_A{}^\mu e_B{}^\nu\nabla_\mu n_\nu.
\]
With this convention, Israel's formula, or the second junction condition, is
\[
S^A{}_B
=-\frac1{8\pi}
\left(
[K^A{}_B]-\delta^A{}_B[K]
\right),
\]
where $[Q]=Q_+-Q_-$.  

\begin{lemma}[Static spherical junction]\label{lem:static-shell}
Let $F$ be continuous and positive at $r=c$ and piecewise $C^1$ across $c$.  On the timelike surface $r=c$, with proper time $d\tau=\sqrt{F(c)}dt$, one has
\[
K^A{}_B
=
\diag
\left(
\frac{F'}{2\sqrt F},
\frac{\sqrt F}{r},
\frac{\sqrt F}{r}
\right).
\]
Consequently, if only $F'$ jumps,
\[
S^A{}_B=\diag(0,p_c,p_c),
\qquad
p_c=\frac{[F']_c}{16\pi\sqrt{F(c)}}.
\]
\end{lemma}

\begin{proof}
The unit normal is
\[
n=\sqrt F\,\partial_r.
\]
Using the shell tangent vectors and the Christoffel symbols
\[
\Gamma^r{}_{tt}=\frac12FF',
\qquad
\Gamma^r{}_{\theta\theta}=-rF,
\qquad
\Gamma^r{}_{\phi\phi}=-rF\sin^2\theta,
\]
one obtains the displayed components of $K^A{}_B$.  Continuity of $F$ gives zero jump in the two angular components, so
\[
[K]=[K^\tau{}_\tau]
=\frac{[F']_c}{2\sqrt{F(c)}}.
\]
Substitution into Israel's formula gives the result.
\end{proof}

Since
\[
F'= -\frac{2m'}r+\frac{2m}{r^2},
\]
continuity of $m$ gives
\[
[F']_c=-\frac{2}{c}[m']_c.
\]
At the inner surface,
\[
[m']_a=\beta a^2,
\qquad
[F']_a=-2\beta a.
\]
At the outer surface,
\[
[m']_{R_2}=-\beta R_2^2,
\qquad
[F']_{R_2}=2\beta R_2.
\]

\begin{theorem}[Israel surface tensors]\label{thm:israel}
The two surface tensors are
\[
S^A{}_B\big|_{r=a}
=
\diag(0,p_a,p_a),
\qquad
p_a
=-\frac{\beta a}{8\pi\sqrt{1-2M_1/a}},
\]
and
\[
S^A{}_B\big|_{r=R_2}
=
\diag(0,p_{R_2},p_{R_2}),
\qquad
p_{R_2}
=\frac{\beta R_2}{8\pi\sqrt{1-2M_2/R_2}}.
\]
Both layers have zero surface energy density.  The inner layer has negative tangential pressure, while the outer layer has positive tangential pressure.
\end{theorem}

To compare invariant and coordinate distributions, let $\delta_{\Sigma_c}$ be characterized by
\[
\int \varphi\,\delta_{\Sigma_c}\,dV_g
=
\int_{\Sigma_c}\varphi\,dV_{\Sigma_c}.
\]
Since
\[
dV_g=r^2\sin\theta\,dt\,dr\,d\theta\,d\phi
\]
and
\[
dV_{\Sigma_c}=c^2\sqrt{F(c)}\sin\theta\,dt\,d\theta\,d\phi,
\]
one has
\[
\delta_{\Sigma_c}=\sqrt{F(c)}\,\delta(r-c).
\]
Therefore
\[
8\pi p_c\delta_{\Sigma_c}
=
\frac12[F']_c\delta(r-c).
\]

\begin{corollary}[Exact agreement of the singular tensors]\label{cor:Israel-match}
The singular part of the distributional Einstein tensor is
\[
G[g_*]_{\mathrm{sing}}
=
8\pi
\left(
S_a\delta_{\Sigma_a}
+
S_{R_2}\delta_{\Sigma_{R_2}}
\right).
\]
Thus the direct distributional limit and the Israel junction calculation agree exactly.
\end{corollary}

This agreement is expected from the general theory \cite{MansouriKhorrami1996,LeFlochMardare2007}; its role here is to identify the singular geometry selected by the variational problem.

\begin{definition}[Negative tangential NEC measure]
{For a smooth admissible profile, define the positive Radon measure $\Xi_u$ by}
\[
d\Xi_u(r)
:=
[-(\rho+p_t)]_+\,dV
=
\frac{\Delta M}{2}
\frac{r^3(u')_+}{\sqrt{F_u(r)}}\,dr.
\]
Its total mass is
\[
\Xi_u([R_1,R_2])
=
\frac{\Delta M}{2}\I[u].
\]
\end{definition}

\begin{theorem}[Concentration and sharp surface cost]\label{thm:concentration}
For every smooth minimizing sequence,
\[
\Xi_{u_n}
\weakstar
\frac{\Delta M}{2}\lambda_*\delta_a.
\]
Moreover,
\[
\frac{\Delta M}{2}\lambda_*
=
4\pi a^2(-p_a).
\]
Hence the sharp smooth variational cost equals the integrated negative Israel surface pressure of the selected inner layer.
\end{theorem}

\begin{proof}
{Let $\varphi\in C([R_1,R_2])$.  By definition,}
\[
\int\varphi\,d\Xi_{u_n}
=
\frac{\Delta M}{2}
\int
\varphi(r)
\frac{r^3}{\sqrt{F_{u_n}(r)}}
\,d\mu_n(r),
\qquad
\mu_n:=(u_n')_+\,dr.
\]
{Theorem~\ref{thm:minseq} gives $\mu_n\weakstar H_*\delta_a$.  Corollary~\ref{cor:metric} gives uniform convergence}
\[
\frac{r^3}{\sqrt{F_{u_n}(r)}}
\longrightarrow
\frac{r^3}{\sqrt{F_*(r)}}.
\]
{To combine these facts, add and subtract the limiting weight.  The error is bounded by the uniform difference of the weights times $\mu_n(\R)$, and these masses converge to $H_*$.  Therefore the error tends to zero.  For the remaining term, weak-* convergence gives}
\[
\int
\varphi(r)
\frac{r^3}{\sqrt{F_*(r)}}
\,d\mu_n(r)
\longrightarrow
H_*\varphi(a)
\frac{a^3}{\sqrt{F_*(a)}}.
\]
{This proves}
\[
\Xi_{u_n}
\weakstar
\frac{\Delta M}{2}
H_*\frac{a^3}{\sqrt{F_*(a)}}\delta_a.
\]

{Now $P_*(a)=0$, because the box profile vanishes to the left of $a$.  Hence}
\[
F_*(a)=1-\frac{2M_1}{a}.
\]
{Using $H_*=1/A(a)$ and the definition of $\Psi$, the coefficient becomes}
\[
\frac{\Delta M}{2}
H_*\frac{a^3}{\sqrt{1-2M_1/a}}
=
\frac{\Delta M}{2}
\frac{a^3}{A(a)\sqrt{1-2M_1/a}}
=
\frac{\Delta M}{2}\Psi(a)
=
\frac{\Delta M}{2}\lambda_*.
\]
{Thus the smooth negative-NEC measures converge to a single Dirac mass whose total mass is exactly the sharp variational value in physical units.}

{Finally, Theorem~\ref{thm:israel} gives}
\[
-p_a
=
\frac{\beta a}{8\pi\sqrt{1-2M_1/a}},
\qquad
\beta=\Delta M H_*.
\]
{Multiplying by the area $4\pi a^2$ of the inner sphere gives}
\[
4\pi a^2(-p_a)
=
\frac{\Delta M}{2}
H_*\frac{a^3}{\sqrt{1-2M_1/a}}
=
\frac{\Delta M}{2}\lambda_*.
\]
{The equality therefore compares two independently defined quantities: the left side is the invariant integrated surface pressure obtained from the Israel tensor, whereas the right side is the infimum of the smooth proper-volume NEC functional. Both sides have the same mass/energy dimension in geometric units and represent the same limiting cost.}
\end{proof}

{\begin{remark}[Energy conditions in the selected geometry]\label{rem:energy-conditions}
For a surface stress tensor $S^A{}_B = \operatorname{diag}(-\sigma,p,p)$ the shell
null, weak, and dominant conditions read $\sigma + p \ge 0$; $\sigma \ge 0$ and
$\sigma + p \ge 0$; and $\sigma \ge 0$ and $|p| \le \sigma$. The selected geometry
stands as follows.

In the bulk $(a,R_2)$ one has $\rho = \beta/(4\pi) > 0$ and $p_r = p_t = -\rho$, so
the null, weak, and dominant conditions hold, the first two marginally and the
third with equality; the strong condition fails, as it does for any $p=-\rho$
region. The exterior regions are vacuum.

On the layers, Theorem~\ref{thm:israel} gives $\sigma = 0$ at both $r=a$ and
$r=R_2$. The inner layer has $p_a < 0$ and therefore violates the shell null and
weak conditions, while the outer layer has $p_{R_2} > 0$ and satisfies both. The
dominant condition, however, requires $|p| \le \sigma = 0$ and so fails at
\emph{both} layers. Thus the concentration statement of Theorem~\ref{thm:concentration}
is specific to the tangential null condition: the negative tangential NEC localizes
entirely on the inner layer, whereas DEC violation does not localize.

The vanishing of $\sigma$ is a property of the ansatz rather than of the optimizer.
Since $m_u$ is an integral of an $L^1$ function, $F_u$ is continuous for every
admissible profile, so $[K^\theta{}_\theta] = [K^\phi{}_\phi] = 0$ and hence
$S^\tau{}_\tau = 0$ at any junction radius. No configuration in this class can carry
a massive layer, and the DEC status of the limit is therefore inherited from the
gauge class, not selected by minimization.
\end{remark}}

\section{Scope, further consequences, and conclusion}

Several consequences follow directly from the explicit optimizer. Since
\[
  x_0<\frac{7M_1}{3},
\]
the interior phase can coexist with the admissibility condition $R_1>2M_2$
only if
\[
  2M_2<\frac{7M_1}{3},
\]
or equivalently
\[
  \frac{\Delta M}{M_1}<\frac16 .
\]
Thus the interior phase, and consequently the square-root stability regime of
Theorem~\ref{thm:stability}, occurs only in a restricted part of parameter space.

For fixed $M_1$, $M_2$ and $R_1$, as $R_2\to\infty$ one has
\[
  x_0\longrightarrow\frac{7M_1}{3},
  \qquad
  H_*=O(R_2^{-3}),
  \qquad
  \lambda_*=O(R_2^{-3}).
\]
Hence a sufficiently broad transition region can make the sharp integrated
defect small, in contrast with the reciprocal-width divergence of
Theorem~\ref{thm:thin}. For fixed $M_1$, $R_1$ and $R_2$, the normalized value
$\lambda_*$ is independent of $\Delta M$, whereas the physical cost
\[
  \frac{\Delta M}{2}\,\lambda_*
\]
vanishes linearly as $\Delta M\downarrow0$.

{The exact solvability of the problem rests on a simple structural feature of
the optimizer. All positive variation is concentrated at a radius $a$ where no
additional mass has yet accumulated,
\[
  P_*(a)=0,
\]
and consequently
\[
  W_{u_*}(a)=q(a),
\]
so the nonlocal part of the weight is inactive at the minimizing jump. The
sharp problem therefore reduces to the one-dimensional minimization of
\[
  \Psi(r)=\frac{q(r)}{A(r)} .
\]
This also explains why $\lambda_*$ depends on $M_1$, $R_1$ and $R_2$, but not
directly on $M_2$. The bulk region carries the entire mass increment,
\[
  \int_a^{R_2}4\pi r^2\rho\,dr=\beta A(a)=\Delta M,
\]
while both surface layers have zero surface density.

The model is deliberately restricted. The ansatz fixes the areal radius and
imposes
\[
  g_{tt}g_{rr}=-1,
\]
which forces
\[
  p_r=-\rho
\]
and therefore saturates the radial NEC. The obstruction in Proposition~\ref{thm:nogo} is
correspondingly elementary; the main content of the paper lies in the sharp
lower bound, the rigidity of near-minimizers, and the selected bulk--layer
geometry.

The functional integrates the pointwise negative part
\[
  \bigl[-(\rho+p_t)\bigr]_+ ,
\]
so regions of NEC violation cannot be canceled by regions where
$\rho+p_t>0$. It does not penalize large positive pressures, violations of the
dominant energy condition, or incompatibility with a prescribed matter
equation of state. In particular, both selected surface layers have $\sigma=0$
and $p\neq0$, and therefore violate the dominant energy condition; the inner
layer also violates the shell NEC. These restrictions make the variational
problem exactly solvable but limit its physical interpretation.

A natural extension is to introduce an independent lapse,
\[
  g=-e^{2\Phi(r)}F(r)\,dt^2+F(r)^{-1}\,dr^2+r^2\,d\Omega^2,
\]
and optimize jointly over $m$ and $\Phi$. One may then ask whether a positive
lower bound persists when both radial and tangential null directions are
included.

A second direction is the joint thin-layer and near-horizon regime
\[
  R_1\downarrow 2M_2,
  \qquad
  R_2-R_1\downarrow 0 .
\]
In this limit the uniform estimate
\[
  F_u(r)\ge\kappa,
  \qquad
  \kappa=1-\frac{2M_2}{R_1},
\]
degenerates, and so do the constants in the uniform continuity and stability
estimates. The selected weight at the inner layer, however, satisfies
\[
  F_*(a)=1-\frac{2M_1}{a},
\]
which remains positive as $a\downarrow 2M_2>2M_1$. Thus the near-horizon limit
introduces no additional singular factor at the optimizer, although the
reciprocal-width divergence of Theorem~\ref{thm:thin} remains.

It would also be useful to obtain stability estimates that interpolate
uniformly between the square-root and linear regimes as $R_1$ approaches the
transition radius $x_0$. These extensions would test whether the variational
selection mechanism persists beyond the one-function Schwarzschild-gauge
setting.}

{
\section*{Statements and declarations}

\paragraph{Funding.} The authors declare that no funds, grants, or other support
were received during the preparation of this manuscript.

\paragraph{Competing interests.} The authors have no competing interests to
declare that are relevant to the content of this article.

\paragraph{Data availability.} No datasets were generated or analyzed during the
current study.

\paragraph{Use of AI tools.} During the preparation of this work the authors used
generative AI assistants (ChatGPT-5.6 Sol, OpenAI; Claude Opus 4.8, Anthropic) for
language editing and \LaTeX{} preparation, and as interactive assistants for
exploratory discussion and for cross-checking intermediate calculations. All
definitions, theorem statements, and proofs are the authors' own. The authors have
independently verified every mathematical argument and every cited reference,
reviewed and edited all content, take full responsibility for it, and confirm that
no AI system is an author or meets authorship criteria.}

\bibliographystyle{unsrtnat}
\bibliography{Sharp_Tangential_NEC_Minimization_and_Israel_Surface_Layers_references}

\end{document}